\documentclass{elsarticle}
\usepackage{amsmath}
\usepackage{amsthm}
\usepackage{amssymb}
\usepackage{bm}
\usepackage{enumitem}
\usepackage{framed}
\usepackage{color}
\usepackage{comment}
\usepackage{relsize}

\usepackage[utf8]{inputenc}

\newtheorem{theorem}{Theorem}[section]
\newtheorem{corollary}[theorem]{Corollary}
\newtheorem{definition}[theorem]{Definition}
\newtheorem{lemma}[theorem]{Lemma}

\newtheorem{observation}[theorem]{Observation}

\newcommand{\Oh}{\mathcal{O}}

\def\II{\mathcal{I}}

\def\NN{\mathbb{N}}
\def\RR{\mathbb{R}}

\newcommand{\STAB}{\mathop{\mathrm{STAB}}}
\newcommand{\FORM}{\mathop{\mathrm{FOP}}}

\newcommand{\CUT}{\mathop{\mathrm{CUT}}}
\newcommand{\PLC}{\mathop{\mathrm{PLC}}}

\newcommand{\conv}{\mathop{\mathrm{conv}}}
\newcommand{\size}{\mathop{\mathrm{size}}}

\newcommand{\xc}{\mathop{\mathrm{xc}}} 

\let\gradd=\triangledown

\begin{document}

\begin{frontmatter}
\title{Parameterized Extension Complexity of Independent Set
		and Related Problems}

\author[jgph]{Jakub Gajarsk\'y\,\fnmark[g1,gacr-ph]}
\ead{jakub.gajarsky@tu-berlin.de}
\fntext[g1]{Current affiliation: Technical University Berlin}

\author[jgph]{Petr Hlin\v{e}n\'y\,\fnmark[gacr-ph]}
\ead{hlineny@fi.muni.cz}

\author[hrt]{Hans Raj Tiwary\,\fnmark[gacr-hrt]}
\ead{hansraj@kam.mff.cuni.cz}

\address[jgph]{Faculty of Informatics, Masaryk University, Brno, Czech Republic.}
\address[hrt]{KAM/ITI, MFF, Charles University, Prague, Czech Republic.}

\fntext[gacr-ph]{Supported by the Czech Science Foundation,
	research projects 14-03501S and (since 2017) 17-00837S.}
\fntext[gacr-hrt]{Supported by the Czech Science Foundation,
	research project GA15-11559S.}



\begin{abstract} 
Let $G$ be a graph on $n$ vertices and $\STAB_k(G)$ be the 
convex hull of characteristic vectors of its independent sets of size at most $k$. 
We study extension complexity of $\STAB_k(G)$ with respect to a fixed
parameter $k$ (analogously to, e.g., parameterized computational complexity of problems). 
We show that for graphs $G$ from a class of bounded expansion it holds
that $\xc(\STAB_k(G))\leqslant \Oh(f(k)\cdot n)$ where the function $f$ depends only on the class. 
This result can be extended in a simple way to a wide range of similarly 
defined graph polytopes.
In case of general graphs we show that there is {\em no function $f$}
such that, for all values of the parameter $k$ and for all graphs on $n$ vertices, 
the extension complexity of $\STAB_k(G)$ is at most $f(k)\cdot n^{\Oh(1)}.$
While such results are not surprising since
it is known that optimizing over $\STAB_k(G)$ is $FPT$
for graphs of bounded expansion and $W[1]$-hard in general,
they are also not trivial and in both cases stronger than the
corresponding computational complexity results.
\end{abstract}
\end{frontmatter}
\vfill

\section{Introduction}
Polyhedral (aka LP) formulations of combinatorial problems belong to the
basic toolbox of combinatorial optimization.
In a nutshell, a set of feasible solutions of some problem is suitably encoded by
a set of vectors, whose convex hull forms a polytope over which one can then
optimize using established tools.
A polytope $Q$ is said to be an {\em extended formulation} or {\em extension} of a
polytope $P$ if $P$ is a projection of $Q$.  Measuring the size of a
polytope by the minimum number of inequalities required to describe it, one
can define the extension complexity of a polytope to be the size of the smallest
extension of the polytope.  This notion has a rich history in combinatorial
optimization where by adding extra variables one can sometimes obtain
significantly smaller polytopes.  For some recent survey on extended
formulations in the context of combinatorial optimization and integer
programming see
\cite{ConfortiCornuejolsZambelli10,Kaibel11,VanderbeckWolsey2010,Wolsey11}.

Since linear (or indeed convex) optimization of a polytope $P$ can instead be
indirectly done by optimizing over an extended formulation of $P$,
this concept provides a powerful model for solving many combinatorial problems. 
Various Linear Program (LP) solvers exist today that perform quite well in
practice and it is desirable if a problem can be modeled as a small-sized
polytope over which one can use an existing LP solver for linear optimization.  
However, in recent years super-polynomial lower
bound on the  extension complexity of polytopes associated with many combinatorial problems have been established.  These bounds have been generalized to various settings, such
as convex extended formulations, approximation algorithms, etc.  These
results are too numerous for a comprehensive listing, but we refer the
interested readers to some of the landmark papers in this regard
\cite{BraunFPS15,CLRS13,FMPTW15,Roth:14}.

Many of the recent lower bounds on the extension complexity of various combinatorial 
polytopes mimic the computational complexity of the underlying problem. For 
example, it is known that the {\em extension complexities of polytopes} related to 
various NP-hard problems are super-polynomial \cite{AT2013,BraunFPS15,FMPTW15,PV13}. 
One satisfying feature of these lower bounds is that they are 
independent of traditional complexity-theoretic assumptions
such as $P\not=NP$.
Though, there also exist polytopes corresponding to polynomial time solvable 
optimization problems whose extension complexity 
is super-polynomial. In particular, the perfect matching polytope was shown to 
have super-polynomial extension complexity by Rothvo{\ss} \cite{Roth:14}. 
Hence even if the extension complexity of a problem mimics its computational
complexity, lower and upper bounds on the former do {\em not\/} follow from the
corresponding computational complexity bounds and constitute nontrivial new
results of independent interest.

One can naturally ask the related questions in the realm of 
{\em parameterized complexity theory}.
In this rapidly grown field each problem instance comes additionally equipped with 
an integer parameter, and the ``efficient'' class denoted by FPT 
({\em fixed-parameter tractable})
is the one of problems solvable, for every fixed value of the parameter,
in polynomial time of degree independent of the parameter.
See Section~\ref{sec:preliminaries} for details.

Similarly as parameterized complexity provides a finer resolution of
algorithmic tractability of problems, parameterized extension complexity
can provide a finer resolution of extension complexities of polytopes of the problems.
We similarly say that a polytope has an {\em FPT extension}
if it has an extension which is,
for every fixed value of the parameter,
of polynomial size with degree independent of the parameter.
Again, see Section~\ref{sec:preliminaries} for details.

We follow this direction of research with a case study of the
{\em independent-set polytope} of a graph, naturally parameterized by the solution size.
We confirm that the extension complexity of the independent-set polytope 
indeed mimics the parameterized computational complexity of the underlying
independent set problem---a finding which is again not implied by the parameterized 
complexity status of this problem and which is actually a lot stronger than
previous related complexity knowledge.
Precisely, we prove:
\begin{itemize}
\item that the independent-set polytope cannot
have an FPT extension for all graphs,
independently of any computational-complexity assumptions
 (Section~\ref{sec:lowerbound}), but
\item linear-sized FPT extensions of the independent-set polytope do exist
on every graph class of bounded expansion (Section~\ref{sec:upperbound}).
\end{itemize}

Seeing the latter result, one may naturally think whether analogous
results hold for other similar problems.
For example, one may consider the polytope of (induced) subgraphs isomorphic to
a given graph $F$, parameterized by the size of~$F$.
Or, more generally, polytopes defined by solutions of non-local problems,
such as the polytope of dominating sets of a certain size.
While ad-hoc adaptations of our technique to such problems are surely possible, 
we prefer to give a ``metatheorem''---a generic solution aimed at all
problems defined in a certain framework.

Namely, we further formulate and prove the following generalizations:
\begin{itemize}
\item there is a natural way to assign a definition of a polytope to every
graph problem expressible in {FO logic},
and these polytopes have linear-sized FPT extensions on every graph class of
bounded expansion when parameterized by the size of the formula expressing
the problem (Section~\ref{sec:generalizations}),
\item for a restricted fragment of FO graph problems,
near-linear-sized FPT extensions of the polytopes exist even on so called
nowhere dense graph classes (Section~\ref{sec:nowheredense}).
\end{itemize}
%
We conclude the paper with some further thoughts and suggestions in
Section~\ref{sec:conclusions}.

%

\section{Preliminaries}
\label{sec:preliminaries}

We follow standard terminology of graph theory and consider finite simple
undirected graphs.
We refer to the vertex and edge sets of a graph $G$ as $V(G)$ and $E(G)$,
respectively.
An {\em independent set} $X$ of vertices of a graph is such
that no two elements of $X$ are adjacent.
By a {\em cut} in a graph $G$ we mean an edge cut, that is,
an inclusion-wise minimal set of edges $C\subseteq E(G)$ such that
$G\setminus C$ has more connected components than~$G$.

For fundamental concepts of parameterized complexity we refer the readers,
e.g., to the monograph~\cite{DBLP:series/txcs/DowneyF13}.
Here we just very briefly recall the needed notions.
Considering a problem $\cal P$ with input of the form
$(x,k)\in \Sigma^*\times \mathbb N$ (where $k$ is a {\em parameter}),
we say that $\cal A$ is {\em fixed-parameter tractable} (shortly FPT)
if there is an algorithm solving $\cal A$ in time
$f(k)\cdot n^{\Oh(1)}$ where $f$ is an arbitrary computable function.
In the (parameterized) {\em$k$-independent set problem}
the input is $(G,k)$ where $G$ is a graph and $k\in\mathbb N$,
and the question is whether $G$  has an independent set of size at least~$k$.

There is no known FPT algorithm for the $k$-independent set problem
in general and, in fact, the theory of parameterized
complexity~\cite{DBLP:series/txcs/DowneyF13} defines complexity
classes $W[t]$, $t\geq 1$, such that the $k$-independent set problem is
complete for $W[1]$.
Problems that are $W[1]$-hard do not admit an FPT algorithm unless the
Exponential Time Hypothesis fails.

Returning back to graph structure, we shall deal with the concept of 
{\em treewidth} of a graph.
Given a graph~$G$, a \emph{tree-decomposition of $G$} is an ordered pair
$(T,\mathcal{W})$, where~$T$ is a tree and 
$\mathcal{W} = \{W_x \subseteq V(G) \mid x \in V(T)\}$ is a collection
of {\em bags} (vertex sets of~$G$), 
such that the following hold:
\begin{enumerate}\parskip-1ex
    \item $\bigcup_{x \in V(T)} W_x = V(G)$;
    \item for every edge~$e = uv$ in~$G$, there exists~$x \in V(T)$ such that~$u,v \in W_x$;
    \item for each~$u \in V(G)$, the set~$\{x \in V(T) \mid u \in W_x\}$
    induces a subtree of $T$.
\end{enumerate}
The {width} of this tree-decomposition is $\max_{x \in V(T)}|W_x|-1$,
and the treewidth $tw(G)$ of~$G$ is the smallest width of a tree-decomposition of~$G$.

It is worth to note that computing an optimal tree-decomposition of a graph $G$
is linear-time FPT with the parameter~$tw(G)$ \cite{bod96}.

\subsection{Fixed-parameter extension complexity}

The {\em size} of a polytope $P$, denoted by $\size(P)$, is defined to be 
the number of facets of $P$, which is the
minimum number of inequalities needed to describe $P$ if it is full-dimensional.
A polytope $Q$ is called an {\em extension} of a polytope $P$ if $P$ can be 
obtained as a linear projection of $Q$. As a shorthand we will say that in 
this case $Q$ is an EF of~$P$. 
As noted in the Introduction, the following is a useful notion:

\begin{definition}[Extension complexity]
The {\em extension complexity of a polytope $P$}, denoted by $\xc(P)$, is 
defined to be the size of the smallest extension. More precisely, 
$$\xc(P):= \min_{Q \text{ an EF of } P}\size(Q).$$
\end{definition}

In the context of {\em fixed-parameter extension complexity}, 
we deal with families of polytopes $\mathcal{P}_n$ where $n\in\mathbb N$,
and a parameter $k$. 
For example, for the independent set 
problem parameterized by a nonnegative integer $k$, 
the family $\mathcal{P}_n$ could be the family of $k$-independent set polytopes 
(cf.~Subsection~\ref{sub:ispolytope}) for a family of $n$-vertex graphs.
The prime question is whether there exists a computable function $f$
such that $\xc(P)\leq f(k)\cdot n^{\Oh(1)}$ for all~$k,n$ and all
$P\in\mathcal{P}_n$.
As a shorthand we will say in the affirmative case that the collection of 
families $\{\mathcal{P}_n: n\in\mathbb N\}$ has {\em FPT extension
complexity} (in a natural analogy with the aforementioned FPT complexity class---it
also readily follows that problems with FPT extension complexity can be solved in FPT time if the extension can be efficiently constructed).

Buchanan in a recent article \cite{B:15} studied the fixed-parameter 
extension complexity of the $k$-vertex cover problem,
and proved that for any graph $G$ with $n$ vertices, 
the $k$-vertex cover polytope has extension complexity at most $\Oh(c^kn)$ 
for some constant $c<2$. Hence this is a nontrivial example of a polytope class with 
FPT extension complexity. Buchanan also raised the question 
whether the $k$-independent set polytope (Definition~\ref{def:ispolytope}) 
admits an FPT extension. 
We answer this in the negative in Theorem~\ref{thm:stab_xc_lb}.
Note that our negative answer does {\em not rely} on any complexity theoretical
assumptions (such as $FPT\not=W[1]$).

We also look at the positive side of the $k$-independent set problem.
It is known that this problem admits an FPT algorithm (w.r.t.~$k$)
on quite rich restricted graph classes, e.g., on classes of bounded
expansion~\cite{NOdM06} (see Subsection~\ref{sub:sparsity} for the definition).
While this finding, in general, does not imply anything about the extension
complexity of the $k$-independent set polytope,
we manage to apply the tools of~\cite{NOdM06} in our setting, and confirm
 -- in Theorem~\ref{thm:bdexptoxc} -- FPT extension complexity of the
$k$-independent set polytope on any graph class of bounded expansion.
We also study a meta-generalization of this result
(to all FO-definable problems) in Section~\ref{sec:generalizations},
and partly generalize our result to nowhere dense graph classes in
Section~\ref{sec:nowheredense}.

\smallskip
In the course of proving aforementioned Theorems~\ref{thm:stab_xc_lb}
and~\ref{thm:bdexptoxc},
we are going to use the following established results on the topic of
extension complexity.

\begin{theorem}[Balas \cite{balas21}]
	\label{thm:union_balas}
	Let $P_1,P_2,\ldots,P_s$ be polytopes and let $P:=\conv(\bigcup_{i=1}^s P_i)$.
	Then, $\xc(P)\leqslant s+\sum_{i=1}^s\xc(P_i).$
\end{theorem}

For a graph $G$, a cut vector is a $0/1$ vector of   
length $|E(G)|$ whose coordinates correspond to whether an edge of $G$
is in a cut $C\subseteq E(G)$ or not.
The {\em cut polytope} is then the convex hull of all the cut vectors of~$G$.
Our negative result relies on the following lower bound.

\begin{theorem}[Fiorini et al.\ \cite{FMPTW15}]\label{thm:xc_cut}
	The extension complexity of the cut polytope of the complete graph
	$K_n$ on $n$ vertices is $2^{\Omega(n)}.$
\end{theorem}


\subsection{The $k$-independent set polytope}
\label{sub:ispolytope}

Let $G$ be a graph on $n$ vertices. Every subset of vertices of $G$ can 
be encoded as a characteristic vector of length $n$. That is, for a subset $S
\subseteq V$, define the characteristic vector $\chi^S$ as follows: $$\chi^S_v
=\left\{\begin{array}{ll}1 & \text{ if } v\in S\\ 0 & \text{ otherwise }\end{array}
\right.$$
\begin{definition}[Independent set polytope]\label{def:ispolytope}
The {\em$k$-independent set polytope of $G$}, denoted by $\STAB_k^\leqslant(G)$, is defined 
to be the convex hull of the characteristic vectors of every independent set of 
size  at most $k$. That is, $$\STAB\nolimits_k^\leqslant(G)=\conv\left(\left\{\chi^S\in\{  0,1\}
^n |\> S\subseteq V \text{ is an independent set of } G; |S|\leqslant k\right\}\right). 
$$ 
In case that $k=n$ we simply speak about the {\em independent set polytope
of~$G$}\,; $\STAB(G)$.
\end{definition}

Alternatively, one could define the polytope $\STAB_k(G)$ to be the convex 
hull of all independent sets of size exactly equal to $k$. That is, $$
\STAB\nolimits_k(G)=\conv\left(\left\{\chi^S\in\{  0,1\} ^n |\> S\subseteq V 
\text{ is an independent set of } G; |S| = k\right\}\right). $$ 
To simplify our situation, we note the following:

\begin{lemma}\label{lem:stabequalnon}
 $\xc(\STAB_k(G))\leqslant \xc(\STAB_k^\leqslant(G))\leqslant k+\sum_{i=0}^k\xc(\STAB_i(G)).$
\end{lemma}
\begin{proof}
Clearly, $\STAB_k(G)$ is a face of $\STAB_k^\leqslant(G)$. 
Therefore, $\xc(\STAB_k(G))\leqslant\xc(\STAB_k^\leqslant(G)).$

On the other hand, $\STAB_k^\leqslant(G)=\conv(\bigcup_{i=1}^{k}\STAB_i(G)),$ and 
therefore $\xc(\STAB_k^\leqslant(G))\leqslant k+\sum_{i=0}^k\xc(\STAB_i(G))$
by Theorem~\ref{thm:union_balas}.
\end{proof}

We would like to remark that the above Lemma~\ref{lem:stabequalnon} shows that any bounds (lower or upper) 
that are valid for $\xc (\STAB_k(G))$ are also asymptotically valid for $\xc(\STAB_k^\leqslant(G)).$ 
Therefore, the notation $\STAB_k(G)$ can be freely used to mean either of the polytopes without
affecting any lower or upper bounds asymptotically.

We shall also use the following result.
\begin{theorem}[Buchanan and Butenko \cite{BB:14}]
	\label{thm:tw_low_xc}
        The extension complexity of a graph's independent set polytope
        is $\Oh(2^{tw} n)$, where $n$ is the number of vertices and 
	$tw$ denotes its treewidth.
\end{theorem}

Note that Buchanan and Butenko give an explicit description of an extension
of the independent set polytope, provided the corresponding tree-decomposition is given.

\subsection{Sparsity and bounded expansion}
\label{sub:sparsity}

A useful toolbox in our research is the theory of sparse graph classes,
largely developed by Ne\v{s}et\v{r}il and Ossona de Mendez.
We follow their monograph~\cite{NOdM12}.

We start by defining the notion of edge contraction. Given an edge~$e = uv$ of a graph~$G$, we let~$G/e$ denote the graph obtained
from~$G$ by \emph{contracting} the edge~$e$, which amounts to deleting the
endpoints of~$e$, introducing a new vertex~$w_{e}$ and making it adjacent to
all vertices in the union of the neighborhoods of $u$ and $v$
(excluding $u,v$ themselves).
%
A \emph{minor} of~$G$ is a graph obtained from a subgraph of~$G$ by contracting 
zero or more edges.
In a more general view, if $H$ is isomorphic to a minor of~$G$, then we call
$H$ a minor of $G$ as well, and we write~$H \preceq G$.

Alternatively, $H$ is a minor of $G$ if there exists a bijection 
$\psi \colon V(H) \rightarrow \{V_1, \ldots, V_p\}$
where $V_1,\ldots,V_p$ are pairwise disjoint subsets of $V(G)$ inducing
connected subgraphs of~$G$,
and $uv\in E(H)$ only if there is an edge in~$G$ 
with an endpoint in each of $\psi(u)$ and $\psi(v)$.
If, moreover, it is required that each subgraph $G[V_i]$ has radius at most~$d$, 
meaning that there exist~$c_i \in V_i$ (a {center}) such that 
every vertex in~$V_i$ is within distance at most~$d$ from $c_i$ in~$G[V_i]$;
then $H$ is called a \emph{shallow minor at depth~$d$} of~$G$
(shortly, a \emph{$d$-shallow minor}).


Note that if $u,v \in V(H)$ in a $d$-shallow minor, and
$u_1\in\psi(u)$ and $v_1\in\psi(v)$, then $d_G(u_1,v_1) \leq (2d+1) \cdot d_H(u,v)$. 
The class of $d$-shallow minors of~$G$ is denoted by $G \gradd d$,
and this is extended to graph classes~$\mathcal{G}$ as well;
 $\mathcal{G} \gradd d = \bigcup_{G \in \mathcal{G}} G \gradd d$. 


One of the most prominent \cite{NOdM12} notions of ``sparsity'' for graph classes
is the following one:
\begin{definition}[Grad and bounded expansion~\cite{NOdM08}]\label{def:bdexpansion}
 	Let $\cal G$ be a graph class. Then the \emph{greatest reduced average density} of $\cal G$
 	with \emph{rank $d$} is defined as 
 	$$
 		\nabla_d(\mathcal{G}) = \sup_{H \in \mathcal{G} \nabla d} \frac{|E(H)|}{|V(H)|}.
 	$$
	A graph class $\mathcal{G}$ has \emph{bounded expansion} if there exists a
	function $f : \mathbb{N} \rightarrow \mathbb{R}$ (called the \emph{expansion function}) such 
	that for all $d \in \mathbb{N}$, $\nabla_d(\mathcal{G}) \leq f(d)$.
\end{definition}

We provide a brief informal explanation of Definition~\ref{def:bdexpansion}.
A graph to be considered ``sparse'' should not, in particular,
contain subgraphs with relatively many edges.
Since $G \gradd 0$ is the set of all subgraphs of $G$, this is captured by
$2\nabla_0(G)$ being the maximum average degree over all subgraphs of~$G$.
However, the definition requires more;
even after contracting edges up to limited depth~$d$,
the resulting shallow minors stay free of relatively dense subgraphs,
with the maximum average degree bounded by $2\nabla_d(\mathcal{G}) \leq2f(d)$.

For example, the class $\cal P$ of all planar graphs has bounded expansion
(even with a constant expansion function).
On the other hand, a class $\cal Q$ obtained from all cliques by subdividing
each edge twice, although also having relatively few edges,
does not have bounded expansion since ${\cal Q}\gradd 1$ contains all graphs.

\section{Lower Bound: Paired Local-Cut Graphs} 
\label{sec:lowerbound}

In this section we deal with a specially crafted graph for a
lower-bound reduction for the $k$-independent set polytope.
We use a shorthand notation $[n]=\{1,2,\dots,n\}$.

\begin{definition}[Paired local-cut graph]\label{def:PLC}
Given positive integers $k$ and~$n$,
let a \emph{Paired Local-Cut Graph}, denoted by $\PLC(k,n)$,
be defined as follows:
\begin{enumerate}
\item We create $k\cdot2^{ \left\lfloor\log{n} \right\rfloor}$ vertices
labeled with tuples $(i,S)$ for $i\in[k]$
and $S\subseteq \left[\left\lfloor{\log{n}}\right\rfloor\right]$.
These vertices will be called \emph{cut vertices}.
Then we create $k(k-1)\cdot2^{2\left\lfloor\log{n}\right\rfloor}$    
vertices labeled with tuples $(i,j,S_1,S_2)$ where $1\leqslant i
\neq j \leqslant k $ and $S,S' \subseteq [{\left\lfloor\log{n} 
\right\rfloor}].$ These vertices will be called \emph{pairing vertices}.
\item \label{it:defPLCcliq}
For each index $i\in[k]$,
we add edges between every pair of cut nodes that have labels $(i,S_1)$
and $(i,S_2)$, where $S_1,S_2 \subseteq
\left[\left\lfloor{\log{n}}\right\rfloor\right]$ are arbitrary.
For each index pair $i,j\in[k]$,
we add edges between every pair of pairing nodes that have labels
$(i,j,S_1,S'_1)$ and $(i,j,S_2,S'_2)$, where $S_1,S_2, S_1',S_2' \subseteq
\left[\left\lfloor{\log{n}}\right\rfloor\right]$ are arbitrary.
\item \label{it:defPLCpair}
Finally, let $u$ be a cut vertex labeled $(i,S)$ and let $v$ be a
pairing vertex labeled $(j_1,j_2,S_1,S_2)$. 
If $i=j_1$ but $S\neq S_1$, then we add the edge $uv$.
Symmetrically, if $i=j_2$ but $S\neq S_2$, then we also add the edge $uv$.
\end{enumerate}
\end{definition}
%
%
%
For ease of exposition we will identify vertices of $\PLC(k,n)$ with their 
labels whenever convenient. 
We first state two easy claims.

\begin{observation}\label{obs:plc_size}
The number of vertices of the graph $\PLC(k,n)$ equals 
{$k(k-1)\cdot2^{2\left\lfloor\log{n}\right\rfloor}
 +k\cdot2^{ \left\lfloor\log{n} \right\rfloor} \leqslant(kn)^2$}. 
\end{observation}

%

\begin{lemma}\label{lem:plc_stab_size} Let $I$ be an independent set in $\PLC(
k,n)$. Then, $|I|\leqslant k^2.$ Moreover, equality holds if and only if  $I$ 
contains exactly one cut vertex for each $1 \leqslant i\leqslant k$ and 
exactly one pairing vertex for each $1\leqslant i \not= j\leqslant k.$ 
\end{lemma}
\begin{proof} By Definition~\ref{def:PLC}.\,\ref{it:defPLCcliq},
the set $I$ can contain at most $k$ cut vertices---at most one 
vertex $(i, S)$ where $S\subseteq \left[\left\lfloor{\log{n}}\right\rfloor\right]$
for each $1\leqslant i \leqslant k.$ Also, $I$ can contain 
at most $k(k-1)=k^2-k$ pairing vertices---at most one vertex $(i,j,S,S')$
 for each ordered pair $1\leqslant i\not=j \leqslant k.$
\end{proof}

In subsequent Lemma~\ref{lem:plc_vs_cuts}, we will
relate the vertices of $\STAB_{k^2}(\PLC(k,n))$ with the vertices 
of the {\em polytope $\CUT(K_r)$} where $r=k\left\lfloor\log{n}\right\rfloor$,
to be defined as follows. 

We group the $r$ vertices of the complete graph $K_r$ into $k$ groups, 
each of size $\left\lfloor\log{n} \right\rfloor$, 
and label the vertices as $v^i_j$ where $1\leqslant i\leqslant k$
    and $1\leqslant j \leqslant \left\lfloor\log{n}\right\rfloor.$ 
We also order the vertices lexicographically according to their labels.
A cut vector of $K_r$, corresponding to a cut $C$, is a $0/1$ vector of 
length $|E(K_r)|={r \choose 2}$  whose coordinates correspond to whether 
an edge of $K_r$ is in the cut $C$ or not. 
The edges of $K_r$ are labeled with pairs $(i
_ 1  ,j _1,i_2,j_2)$ where $1 \leqslant i_1, i_2 \leqslant k~;~ 1\leqslant j_1
 ,j _ 2 \leqslant \left\lfloor\log{n}\right\rfloor,$ and $(i_1,j_1) < (
i _2,j_ 2)$ lexicographically. So, if $z$ is a cut vector corresponding to a 
given cut $C\subset E _r$ , then $z_{i_1,j_1,i_2,j_2}=1$ if and only if 
the edge $(i_1,j_1,i_2,j_2 )$ is in $C$. The polytope $\CUT(K_r)$ 
is the convex hull of all such cut vectors.

\begin{lemma}\label{lem:plc_vs_cuts}  For every pair of natural numbers $(k,n)
 $ and $r=k\left\lfloor\log{n}\right\rfloor$ it holds that $\CUT(K_r)$ is a 
projection of $ \STAB_{k^2}\left(\PLC\left(k,n\right)\right).$  \end{lemma} 

\begin{proof}
Recall that the independent set vectors of $\STAB_{k^2}(\PLC(k,n))$
are of length $s= k(k-1)\cdot2^{2\left\lfloor\log{n}\right\rfloor}
 +k\cdot2^{ \left\lfloor\log{n} \right\rfloor}$, by Observation~\ref{obs:plc_size}.
We describe an affine map $\pi:\RR
^{s} \to \RR^{r\choose 2}$ such that for every vertex $C$ of $\CUT(K_r)$ 
there exists a vertex $I$ of $\STAB_{k^2}(\PLC(k,n))$ such that $\pi(I)=C$. 
Moreover, for every vertex $I$ of  $\STAB_{k^2}(\PLC(k,n))$ we show that $
\pi(I)$ is a vertex of $\CUT(K_r)$. Since $\pi$ is an affine map, this will 
complete the proof.

To make it easy to follow our arguments,
we relate the coordinates of $\RR^s$ to the vertices of $\PLC(k,n)$,
labeling the coordinates with tuples of the form $(i,j,S,S')$ as follows.
For a coordinate corresponding to a cut vertex $(i,S)$ we label this 
coordinate with $(i,i,S,S)$. For a coordinate 
corresponding to a pairing vertex $(i,j,S,S')$ we label the coordinate 
with same $(i,j,S,S')$. 
Similarly, we identify the coordinates of $\RR^{r\choose 2}$ 
with the pairs of vertices of~$K_r$:
the coordinate corresponding to an edge between two distinct vertices 
$v^i_{\ell_1}$ and $v^j_{\ell_2}$ is to be 
labeled with the integer tuple $(i,\ell_1,j,\ell_2)$, 
assuming that $v^i_{\ell_1} < v^j_{\ell_2}$ lexicographically
(that is, $i\leqslant j$ and if $i=j$ then $\ell_1<\ell_2$). 
Also note that $1 \leqslant i,j \leqslant k $  and $1\leqslant  \ell_1,\ell_2 
\leqslant \lfloor\log{n}\rfloor$. 

Given a vector $y\in\RR^s$ we define $\pi(y):=z\in \RR^{r\choose 2}$ 
where $$z_{i_1,\ell_1,i_2,\ell_2}= \left\{\begin{array}{ll}
\mathlarger{\mathlarger{\sum}}\limits_{\substack{\\S\subseteq \left[ \left\lfloor\log
{n} \right\rfloor\right]\\\ell_1\notin S\\\ell_2 \in S}}y_{i_1,i_1,S,S} +
\mathlarger{\mathlarger{\sum}}\limits_{\substack{ S\subseteq \left[ \left\lfloor\log
{n} \right\rfloor\right]\\\ell_1\in S\\\ell_2 \notin S}}y_{i_1,i_1,S,S} & 
		\text{ if } i_1=i_2 \,,
 \\
\mathlarger{\mathlarger{\sum}}\limits_{\substack{ S'\subseteq \left[ 
\left\lfloor\log{n} \right\rfloor\right]\\S''\subseteq \left[ \left\lfloor\log
{n} \right\rfloor\right]\\\ell_1\notin S'\\\ell_2 \in S''}}
		y_{i_1,i_2,S',S''} +
\mathlarger{\mathlarger{\sum}}\limits_{\substack{ S'\subseteq \left[ 
\left\lfloor\log{n} \right\rfloor\right]\\S''\subseteq \left[ \left\lfloor\log
{n} \right\rfloor\right]\\\ell_1\in S'\\\ell_2 \notin S''}}
		y_{i_1,i_2,S',S''} & \text{ if } i_1\neq i_2 \,.
\end{array}\right. 
$$
 
Let $y\in\RR^s$ be a vertex of $\STAB_{k^2}(\PLC(k,n)).$ That is, $y$ is the 
characteristic vector of an independent set $I\in\II$. Since $I$ is of size $k
^2$, for every \mbox{$1\leqslant i\not=j\leqslant k$} the following hold by
Lemma~\ref{lem:plc_stab_size}:
\begin{itemize}
\item there is exactly one
$S_i\subseteq \left[\left\lfloor{\log{n}}\right\rfloor\right]$
such that $y_{i,i,S_i,S_i}=1$, and
\item there is exactly one pair
$S'_{ij},S''_{ij}\subseteq \left[\left\lfloor{\log{n}}\right\rfloor\right]$
such that $y_{i,j,S'_{ij},S''_{ij}}=1$.
\end{itemize}
Furthermore, by Definition~\ref{def:PLC}.\,\ref{it:defPLCpair},
for any pairing vertex $(i_1,i_2,S',S'')$ picked in $I$, that is
$y_{i_1,i_2,S',S''}=1$, it holds;
if $i=i_1$ then $S_i=S'$, and if $i=i_2$ then $S_i=S''$. 
Consider the subsets $S(I), \overline{S}(I)$ of vertices of $K_r$ defined as follows: 
$$
S(I) := \left\{v^i_\ell |\> i \in [k] \wedge \ell \in S_i \right\}
\mbox{ and }\>
\overline{S}(I) := \left\{v^i_\ell |\> i \in [k] \wedge \ell \not\in S_i \right\}
$$
It is routine to check that $\pi(y)$ is exactly the characteristic vector 
of the cut defined by $S(I)$, $\overline{S}(I)$ because $z_{i_1,\ell_1,i_2,\ell_2}=1$ 
if and only if $v^{i_1}_{\ell_1}$ and $v^{i_2}_{\ell_2}$ do not lie both in $S(I
)$ or both in $\overline{S}(I).$

On the other hand, any cut in $K_r$ of characteristic vector~$z$ 
creates a bipartition $(S,\overline{S})$ of the vertices of $K_r$.
The bipartition $(S,\overline{S})$ consequently
induces bipartitions $(S_i,\overline{S}_i)$, $i=1,\dots,k$, 
within each of the $k$ groups of the vertices of $K_r$;
namely $S_i:=\big\{j |\> 1\leqslant j \leqslant \lfloor\log{n}\rfloor
	\wedge$ $ v^i_j\in S\big\}$.
Then $\big\{(i,S_i) |\> {1\leqslant i\leqslant k}\big\}\cup
	\big\{(i,j,S_i,S_j) |\> {1\leqslant i \not= j\leqslant k}\big\}$ 
is an independent set of $\PLC(k,n)$ whose size is $k^2$ 
and whose characteristic vector projects to $z$ under~$\pi.$

Hence $\pi$ defines a projection from $\PLC(k,n)$ to $\CUT(K_r).$
\end{proof}

\begin{corollary}\label{cor:stab_plc_lb} There exists a constant $c'>0$ such 
that for $k,n\in\NN$,
 $$\xc\left(\STAB\nolimits_{k^2}(\PLC(k,n))\right) \geqslant  n^{c{
	'}k}.$$ 
\end{corollary}

\begin{proof}
By Lemma \ref{lem:plc_vs_cuts}, $\STAB_{k^2}(\PLC(k,n))$ is an 
extended formulation of $\CUT(K_r)$ with $r=k\lfloor\log{n}\rfloor$. So any 
extended formulation of $\STAB_{k^2}(\PLC(k,n))$ is also an extended 
formulation of $\CUT(K_r).$ By Theorem \ref{thm:xc_cut}, $\xc\big(\CUT(K_r)\big) \geqslant 2^{\Omega(r)}$. Therefore, $\xc\big(\STAB_{k^2}(\PLC(k,n))\big)
 \geqslant \xc\big(\CUT(K_r)\big) \geqslant 2^{\Omega(r)} \geqslant  n^{c{'}k}$ for 
some constant $c'>0.$ 
\end{proof}

We can now easily finish with the main result of this section.
\begin{theorem}\label{thm:stab_xc_lb}
 There is no function $f:\NN\to\RR$ such that $\xc(\STAB_k(G))\leqslant f(k)
\cdot n^{\Oh(1)}$ for all natural numbers $k$ and all graphs $G$ on $n$ 
vertices.
\end{theorem}

\begin{proof} 
Suppose, on the contrary, that such a function $f$ does exist. That is, there is 
a constant $c$ such that for every pair of natural numbers $(\ell,m)$ and for 
all $m$-vertex graphs $G$ it holds that 
$\xc(\STAB_\ell(G))\leqslant f(\ell)\cdot m^{c}$. 
 
Given a pair $(k,n)$ of natural numbers consider the graph $\PLC(k,n)$.  
By Corollary~\ref{cor:stab_plc_lb}, we have that $\xc\left(\STAB_{k^2}(\PLC(k,n))\right) 
\geqslant  n^{c{'}k}$ for some constant $c'>0$. 
On the other hand, we have $\ell=k^2$ and $m\leqslant(kn)^2$ by
Observation~\ref{obs:plc_size}, and so we derive from our assumption that
$\xc\left(\STAB_{k^2}(G)\right)\leqslant f(k^2)\cdot (kn)^{2c}$. 
However, implied $n^{c{'}k}\leqslant f(k^2)\cdot (kn)^{2c}$ clearly cannot
hold true for a sufficiently large but fixed parameter $k$ and arbitrary~$n$.
More precisely, we can choose $k>2c/c'$ and $n$ such that
$\log{n}>(\log{f(k^2)}+2c\log{k})/(c'k-2c)>0$ and obtain
$$
n^{c{'}k}=n^{2c}\cdot n^{c'k-2c} > 
	n^{2c}\cdot 2^{\log{f(k^2)}+2c\log{k}}=f(k^2)\cdot n^{2c}k^{2c}
,$$
a~contradiction. Hence no such function $f$ exists. 
\end{proof}

We remark that the function $f$ in the previous theorem need not even be computable.

\section{Upper Bound: Bounded Expansion Classes}
\label{sec:upperbound}

While Theorem~\ref{thm:stab_xc_lb} asserts that FPT extensions are not
possible for the $k$-independent set polytopes of all graphs,
there is still a good chance to prove a positive result for restricted
classes of graphs.
An example of such restriction is, by a simple modification of Theorem~\ref{thm:tw_low_xc},
presented by graph classes of bounded treewidth;
although, this is somehow too restrictive.
We show that in the case of $k$ being a fixed parameter, one can go much
further.


The underlying idea of our approach can be informally explained as follows.
Imagine we can ``guess'', in advance, a (short) list of well-structured subgraphs of
our graph such that every possible independent set is fully contained in at
least one of them.
Then we can separately construct an independent set polytope for each one of
the subgraphs, and make their union at the end (Theorem~\ref{thm:union_balas}).
This ambitious plan indeed turns out to be viable for graph classes of
bounded expansion (Definition~\ref{def:bdexpansion}), and the key to the
success is a combination of a powerful structural characterization
of bounded expansion (Theorem~\ref{thm:low_td_coloring})
with the size bound $k$ on the independent sets.

In order to state the desired structural characterization, we need the
notion of treedepth.
In this context, a \emph{rooted forest} is a disjoint union of rooted trees. 
The {height of a rooted forest} is the maximum distance from one of the
forest's roots to a vertex in the same tree. 
The \emph{closure} $\text{clos}(\mathcal{F})$ of a rooted forest~$\mathcal{F}$ 
is the graph with the vertex set 
$\bigcup_{T \in \mathcal{F}} V(T)$ and the edge set $\{xy : \text{$x$ is
  an ancestor of $y$ in a tree of $\cal F$}\}$. 
The \emph{treedepth} $\text{td}(G)$ of a graph~$G$ is the minimum height
plus one of a rooted forest~$\mathcal{F}$ such that $G \subseteq
\text{clos}(\mathcal{F})$.

The following fact, which will allow us to connect the expansion concept
of this section with Theorem~\ref{thm:tw_low_xc},
is easy to establish directly from the definitions:
\begin{observation}\label{obs:tdtw}
For any $G$, the treewidth of $G$ is at most the treedepth of $G$ minus one.
\end{observation}


The amazing connection between graph classes of bounded expansion and 
treedepth is captured by the notion of low treedepth coloring: 
For an integer $d \ge 1$, an assignment of colors to the vertices
of a graph $G$ is a {\em low treedepth coloring of order $d$} if, for every
$s = 1, 2, \dots , d$, the union of any $s$ color classes induces a subgraph 
of $G$ of treedepth at~most~$s$. 

In particular, every low treedepth coloring of $G$ is a
proper coloring of $G$ (but not the other way round),
and the union of any two color classes induces a forest of stars.
The following result is crucial:

\begin{theorem}[Ne\v{s}et\v{r}il and Ossona de Mendez~\cite{NOdM06,NOdM08}]
\label{thm:low_td_coloring}
If $\mathcal{G}$ is a class of graphs of bounded expansion,
then there is a function $N_{\mathcal{G}} : \mathbb{N} \rightarrow \mathbb{N}$ 
(depending on the expansion function of~$\mathcal{G}$)
such that for any graph $G \in \mathcal{G}$ and $k$, 
there exists a low treedepth coloring of order $k$ of $G$ 
using at most $N_{\mathcal{G}} (k)$ colors.
This coloring can be found in linear time for a fixed~$k$.
\end{theorem}

We are now ready to state and prove the main theorem of this section.

\begin{theorem}\label{thm:bdexptoxc}
 Let $\mathcal{G}$ be any graph class of bounded expansion. Then 
there exists a computable function $f:\NN\to\NN$, depending on the expansion function
of~$\mathcal{G}$, such that
$$\xc\!\big(\STAB\nolimits_k(G)\big)\leqslant f(k)\cdot n $$ 
holds for every integer $n$ and every $n$-vertex graph~$G\in\mathcal{G}$.
Moreover, an explicit extension of $\STAB\nolimits_k(G)$ of size at most
$f(k)\cdot n$ can be found in linear time for fixed $k$ and $\mathcal{G}$.
\end{theorem}
\begin{proof}
Since $\mathcal{G}$ is a graph class of bounded expansion, by
Theorem~\ref{thm:low_td_coloring} we can for any $G \in \mathcal{G}$ and
given $k$ find an assignment $c:V(G)\to[N_{\mathcal{G}}(k)]$
such that $c$ is a low treedepth coloring of order~$k$.
Let $\mathcal{J}_k:={[N_{\mathcal{G}}(k)] \choose k}$ denote the set
of $k$-element subsets of~$[N_{\mathcal{G}}(k)]$,
and let a subgraph $G_J\subseteq G$ where $J\in \mathcal{J}_k$, be defined 
as the subgraph of $G$ induced on 
$\bigcup_{j\in\mathcal{J}_k} c^{-1}(j)$~ -- the color classes indexed by~$J$.

Note the following two immediate facts:
\begin{itemize}
\item[a)] by the definition, each $G_J$, $J\in \mathcal{J}_k$, 
	is of treedepth at most~$|J|=k$;
\item[b)] for every set $X\subseteq V(G)$ (independent or not) of size~$|X|\leq k$,
	there is $J\in \mathcal{J}_k$ such that $X\subseteq V(G_J)$.
\end{itemize}
Consequently,
$$\STAB\nolimits_k(G) = \conv\left(
	\bigcup\nolimits_{J\in \mathcal{J}_k} \STAB\nolimits_k(G_J)
	\right)
$$
and it is sufficient to bound the extension complexity of each $\STAB\nolimits_k(G_J)$.

By (a) and Observation~\ref{obs:tdtw}, $tw(G_J)\leq k-1$ and
Theorem~\ref{thm:tw_low_xc} applies here:
$\xc\!\big(\STAB\nolimits_k(G_J)\big)\leqslant \Oh(2^k\cdot|G_J|)
	\leqslant \Oh(2^k\cdot n) \leqslant c'2^k\cdot n$ for a suitable
constant~$c'$.
Then, by Theorem~\ref{thm:union_balas}, we have
\begin{align*}\qquad
\xc\!\big(\STAB\nolimits_k(G)\big) &\leqslant
	|\mathcal{J}_k|+\sum_{J\in\mathcal{J}_k}
		\xc\!\big(\STAB\nolimits_k(G_J)\big)
\\ &\leqslant |\mathcal{J}_k|\cdot(1+c'2^k\cdot n)
	\leqslant {N_{\mathcal{G}}(k) \choose k}(1+c'2^k)\cdot n
	\leqslant f(k)\cdot n
\,.\end{align*}

Note that this extended formulation can be constructed in linear time for
fixed $k$ since the low treedepth coloring
in Theorem~\ref{thm:low_td_coloring} can be found in linear time, the
extended formulation in Theorem~\ref{thm:tw_low_xc} is explicit
using a tree-decomposition trivially derived from the definition of
tree-depth, and the
extended formulation of union of polytopes can be constructed in linear time
from the extensions of the component polytopes~\cite{balas21}.
\end{proof}

\section{Generalizing the Upper Bound}
\label{sec:generalizations}

As advertised in the introduction, the positive results and the proof
method from Section~\ref{sec:upperbound} can be generalized to
the polytopes of many more graph problems than only the 
$k$-independent set polytope.
In this respect one could think about other established problems like, for
example, the $k$-clique or $k$-dominating set.
Instead of giving a list of particular extensions, we provide here a
metaresult covering a whole range of graph problems which share
a common ground with the independent set problem.


The first step in this generalization is introducing our descriptive
framework, namely the first-order logic of graphs,
and the polytopes associated with graphs under a given logical formula.
At this point the reader should understand that defining such a polytope
for a logical formula cannot be as simple as
Definition~\ref{def:ispolytope}, due to necessity to handle formula
arguments in full generality (as they may not be only mutually 
interchangeable elements of a set).
However, as it will be clear in the case of the independent set polytope,
our polytopes defined from logical formulas naturally form extensions of
what one would call ``standard'' problem polytopes.
In particular, an upper bound on the extension complexity of our FO polytope
from Definition~\ref{def:phipolytope} applies also to such ``standard''
problem polytopes of particular graph problems.

\subsection{FO logic and FO polytope}

The {\em first-order logic of graphs} (abbreviated as FO) applies the
standard language of first-order logic to a graph $G$ viewed as a relational 
structure with the domain $V(G)$ and the single binary (symmetric) relation $E(G)$.
That is, in FO we have got the standard predicate $x=y$, a binary predicate
$edge(x,y)$, usual logical connectives $\wedge,\vee,\to$, and quantifiers
$\forall x$, $\exists x$ over the vertex set $V(G)$.
For example, $\phi(x,y)\equiv \exists z\big(edge(x,z)\wedge edge(y,z)\big)$
states that the vertices $x,y$ have a common neighbor in~$G$.

If $\phi$ is a formula of $k$ free variables and
$W=(w_1,w_2,\dots,w_k)\in V(G)^k$ is such that $\phi(w_1,w_2,\dots,w_k)$ holds
true in $G$, we write $G\models\phi(w_1,w_2,\dots,w_k)$.
Consider now the FO formula 
$$\iota(x_1,x_2,\dots,x_k)\equiv
 \bigwedge_{i\not=j}\big(\neg edge(x_i,x_j)\wedge x_i\not=x_j\big)$$
which is quantifier-free.
It is easy to see that $G\models\iota(w_1,w_2,\dots,w_k)$ if and only if
$\{w_1,w_2,\dots,w_k\}$ is an independent set of size exactly~$k$.

In another example,
$$\delta(x_1,x_2,\dots,x_k)\equiv \forall y
 \bigvee_{i=1,\dots,k}\big(edge(x_i,y)\vee x_i=y\big)$$
is an FO formula with one quantifier such that
$G\models\delta(w_1,\dots,w_k)$ if and only if
$\{w_1,\dots,w_k\}$ is a dominating set (of size $\leq k$).
A more involved example is the following formula with two
quantifiers describing a distance-$2$ dominating set:
$$\delta_2(x_1,\dots,x_k)\equiv \forall y \exists z
 \bigvee_{i=1,\dots,k}\big[ (edge(x_i,z)\vee x_i=z)
	\wedge (edge(z,y)\vee z=y) \big]$$
	
For our purposes we will consider FO logic on graphs \emph{labeled} by  labels from
a finite set $Lab$.
Formally, \textit{vertex labels} are modelled as subsets of $V(G)$ -- for each $a \in Lab$ there is a subset $V_a$ of vertices having label $a$. From FO formulas labels are accessed using unary predicates: $L_a(v)$ is true if and only of $v \in V_a$. As an example, if we work with graphs labeled by $Lab = \{a,b\}$ then the formula
$$\delta'(x_1,x_2,\dots,x_k)\equiv  \forall y~ L_a(y) \Rightarrow
 \bigvee_{i=1,\dots,k} \big(edge(x_i,y)\vee x_i=y\big)$$
such that $G\models\delta(w_1,\dots,w_k)$ if and only if all vertices labeled by $a$ are dominated by
$\{w_1,\dots,w_k\}$. Apart from graphs with labeled vertices, one can also consider graphs with labeled edges. \textit{Edge labels} are realized as subsets $E_a$ of $E(G)$ and are accessed by FO formulas using binary predicates $edge_a(x,y)$, i.e. $edge_a(u,v)$ is true if and only if $\{u,v\} \in E_a$. As an example, if we work with graphs with edges labeled by $Lab = \{a,b\}$ then the formula
$$\delta''(x_1,x_2,\dots,x_k)\equiv  \forall y 
 \bigvee_{i=1,\dots,k} \big(edge_a(x_i,y)\vee x_i=y\big)$$
such that $G\models\delta(w_1,\dots,w_k)$ if and only if 
vertices $\{w_1,\dots,w_k\}$ dominte all vertices of $G$ using only edges labeled by label $a$. Edge labels and vertex labels can be used simultaneously -- this is the setting we use in Lemma~\ref{lem:fullFOtoex}.

Now we assign,
to any FO formula $\phi(x_1,x_2,\dots,x_k)$, a graph polytope as follows.
As we have already mentioned above,
it has to be somehow more complicated than the independent set
polytope since the order of arguments of $\phi$ matters in general, 
and the same vertex may be repeated among the arguments.
For an ordered $k$-tuple of vertices $W=(w_1,w_2,\dots,w_k)\in V(G)^k$
we thus define its characteristic vector $\chi^W$ of length $k|V(G)|$ by
$$\chi^W_{v,i} =\left\{\begin{array}{ll}1 & \text{ if } v=w_i ,\\
	 0 & \text{ otherwise.}\end{array}
\right.$$
Note that $\chi^W$ always satisfies $\sum_{v\in V(G)}\chi^W_{v,i}=1$
for each~$i=1,\dots,k$, by the definition.

We can now give the following definition:
\begin{definition}[FO polytope]\label{def:phipolytope}
Let $\phi(x_1,\dots,x_k)$ be an FO formula with $k$ free variables.
The {\em(first-order) $\phi$-polytope of $G$}, denoted by $\FORM_\phi(G)$, is defined 
to be the convex hull of the characteristic vectors of 
every $k$-tuple of vertices of $G$ such that $\phi(w_1,w_2,\dots,w_k)$ holds
true in $G$.
That is, 
\begin{align*}
\FORM\nolimits_\phi(G)=\conv\left(\left\{ \chi^W\in\{0,1\}^n |
 \right.\right. &\> W=(w_1,w_2,\dots,w_k)\in V(G)^k,
\\ &\left.\left. G\models\phi(w_1,w_2,\dots,w_k)^{\vbox to 1ex{}} \right\}\right). 
\end{align*}
\end{definition}

The definition of an FO polytope is, at least in the case of an independent
set problem, indeed very naturally related to Definition~\ref{def:ispolytope}
of the independent set polytope. See the following:
\begin{lemma}
Let $\iota(x_1,\dots,x_k)\equiv
 \bigwedge_{i\not=j}\big(\neg edge(x_i,x_j)\wedge x_i\not=x_j\big)$
(the above $k$-independent set formula).
For every graph $G$,
the $\iota$-polytope $\FORM_\iota(G)$ is an extension of\/ $\STAB_k(G)$.
\end{lemma}
\begin{proof}
If $G$ has $n$ vertices then
$$\STAB\nolimits_k(G)=\left\{y\in\RR^n\left|\> y_v=\sum_{i=1}^k\chi^W_{v,i},~
\chi^W\in\FORM\nolimits_\iota(G)\right.\right\}.$$
Therefore, $\STAB_k(G)$ is a
projection of $\FORM_\iota(G)$ given by the projection map described by
$y_v=\sum_{i=1}^k\chi^W_{v,i}$ for all vertices $v$ of $G$.
\end{proof}

\subsection{Upper bound for existential FO}

For the subsequent arguments we recall the following weaker form%
\footnote{The original result of Kolman et al.  applies
to Monadic Second Order logic: a logic that subsumes FO logic.} 
of a recent result of Kolman et al.~\cite{KKT15}:

\begin{theorem}[Kolman, Kouteck\'y and Tiwary \cite{KKT15}]\label{thm:fo_tw}
Let $\phi(x_1,x_2,\dots,x_k)$ be an FO formula with $k$ free variables
and~$\ell$ quantifiers.
Then there exists a computable function $g:\NN\times\NN\to\NN$, such that
$$\xc\!\big(\FORM\nolimits_\phi(G)\big)\leqslant g(k+\ell,\tau)\cdot n $$ 
holds for every integer $n$ and every $n$-vertex graph~$G$ of treewidth $\tau$.
Furthermore, this extension can be computed in linear time
for fixed $k,\ell$ and~$\tau$.
\end{theorem}

Using this and the decomposition provided by Theorem \ref{thm:low_td_coloring}, 
we are able to directly extend Theorem~\ref{thm:bdexptoxc}
to the following restrictive fragment of FO logic.
We say that an FO formula $\phi(x_1,\dots,x_k)$ is {\em existential FO}
if it can be written as $\phi(x_1,\dots,x_k)\equiv\exists y_1\dots y_\ell
 \,\psi(x_1,\dots,x_k,y_1,\dots,y_\ell)$,
where $\psi$ is quantifier-free.

\begin{lemma}\label{lem:exFOupper}
Let $\phi(x_1,x_2,\dots,x_k)$ be an existential FO formula with $k$ free variables
and $\ell$ quantifiers.
Also, let $\mathcal{G}$ be any graph class of bounded expansion.
Then there exists a computable function $f:\NN\to\NN$, 
depending on the expansion function of~$\mathcal{G}$, such that
$$\xc\!\big(\FORM\nolimits_\phi(G)\big)\leqslant f(k+\ell)\cdot n $$ 
holds for every integer $n$ and every $n$-vertex graph~$G\in\mathcal{G}$.
Furthermore, an explicit extension of $\FORM\nolimits_\phi(G)$ of size at most
$f(k+\ell)\cdot n$ can be found in linear time for fixed $\mathcal{G}$ and $k,\ell$.
\end{lemma}

\begin{proof}
We start with two simple facts from model theory:
\begin{itemize}
\item[a)] If $H$ is an induced subgraph of $G$,
and $H\models\phi(w_1,\dots,w_k)$ for $w_1,\dots,w_k\in V(H)$,
then $G\models\phi(w_1,\dots,w_k)$ (since $\phi$ is existential).
\item[b)] If $G\models\phi(w_1,\dots,w_k)$ for any 
$W=\{w_1,\dots,w_k\}\subseteq V(G)$,
then there is $U\subseteq V(G)$, $|U|\leq\ell$, such that
$G[W\cup U]\models\phi(w_1,\dots,w_k)$ where $G[W\cup U]$ is the subgraph of
$G$ induced on $W\cup U$
(since $\phi$ has~$\leq\ell$ quantifiers).
\end{itemize}

We can hence apply the same technique as in the proof of
Theorem~\ref{thm:bdexptoxc} -- using a low treedepth coloring
$c$ of $G$ now by $N_{\mathcal{G}}(k+\ell)$ colors from
Theorem~\ref{thm:low_td_coloring}.
Again, let $\mathcal{J}_{k+\ell}:={[N_{\mathcal{G}}(k+\ell)] \choose k+\ell}$ denote the set
of $(k+\ell)$-element subsets of~$[N_{\mathcal{G}}(k+\ell)]$,
and let a subgraph $G_J\subseteq G$ where $J\in \mathcal{J}_{k+\ell}$, be defined 
as the subgraph of $G$ induced on 
$\bigcup_{j\in\mathcal{J}_{k+\ell}} c^{-1}(j)$~ -- the color classes of $c$ indexed by~$J$.
By a),b) we immediately get
$$\FORM\nolimits_\phi(G) = \conv\left(
	\bigcup\nolimits_{J\in \mathcal{J}_{k+\ell}} \FORM\nolimits_\phi(G_J)
\right)
.$$

From Theorems~\ref{thm:union_balas} and~\ref{thm:fo_tw}
(via Observation~\ref{obs:tdtw}) we analogously conclude
\begin{align*}\qquad
\xc\!\big(\FORM\nolimits_\phi(G)\big) &\leqslant
	|\mathcal{J}_{k+\ell}|+\sum_{J\in\mathcal{J}_{k+\ell}}
		\xc\!\big(\FORM\nolimits_\phi(G_J)\big)
\\ &\leqslant |\mathcal{J}_{k+\ell}|\cdot\big(1+g(k+\ell,k+\ell-1)\cdot n\big)
	\leqslant f(k+\ell)\cdot n
\,.\end{align*}
Again, this extended formulation can be constructed in linear time for
fixed $k,\ell$.
\end{proof}

\subsection{Extension towards full FO}

Existential FO is a rather restricted fragment as, for example,
the vertex cover or dominating set problems cannot (at least not
immediately) be formulated in it.
Though, using another established logical tool, explained next, 
we can circumvent this restriction and cover problems
in full FO logic of graphs, i.e., allowing
also for universal quantifiers in the problem expression.
We base our approach
on the exposition of an FO model checking algorithm for graphs of 
bounded expansion presented in~\cite{GK09}.

For $i,q\in\NN$, $i\leq q$,
we say that two $i$-tuples of vertices $\bar{u}, \bar{v}\in V(G)^i$ have the same
\textit{logical $q$-type}\footnote{Our logical types correspond to
\emph{full types} (Definition 8.17) in~\cite{GK09}} (denoted by
$tp_i^q(\bar{u}) = tp_i^q(\bar{v})$\,) if they satisfy the same set of FO formulas
with quantifier rank at most $q-i$.  We note that even though there are
infinitely many formulas of a given quantifier rank, there are only finitely
many semantically different ones, and therefore there are only finitely many
$q$-types.
Let $\mathcal{T}_i^q$ denote the finite set of all $q$-types of $i$-tuples of
vertices.

The following lemma (Lemma 8.21 in~\cite{GK09} adjusted to our setting) says
that on graph classes of bounded expansion one can reduce the problem of
determining the $q$-type of an $i$-tuple of vertices of $G$ to evaluating
certain \textit{existential} FO formula on a suitable and efficiently computable
labeling of $G$.

\begin{lemma}[\cite{GK09}]\label{lem:fullFOtoex}
Let $\mathcal{G}$ be a class of graphs of bounded expansion, and let $q \ge 0$.
There exists $r := r(q, \mathcal{G}) \in \mathbb{N}$ and a finite set
$\text{Lab}_r$ of special labels such that the following holds 
for all $ 1 \le i \le q$:
there are existential first-order formulas
$\psi_t(x_1,\ldots,x_i)$ for $t \in \mathcal{T}_i^q$, 
using labels from $\text{Lab}_r$, 
such that for every graph $G \in \mathcal{G}$
and a low treedepth coloring $c$ of $G$ of order~$r$
(using $N_{\mathcal{G}}(r)$ colors),
it is possible to efficiently (in polynomial time for fixed $\mathcal{C}$ and $q$) label vertices and edges of $G$ using $c$ and labels
from $\text{Lab}_r$, to get a labeled graph $G(c)$
such that for every tuple $\bar{v} \in V(G)^i$ it holds
$$ G(c) \models \psi_t(v_1,\ldots,v_i)
 \quad \text{if and only if} \quad tp_i^q(\bar{v})=t \text{ in $G$}.$$
\end{lemma}

%

Now we state the final strengthening of Lemma~\ref{lem:exFOupper}.

\begin{theorem}\label{lem:fullFOupper}
Let $\phi(x_1,\dots,x_k)$ be an FO formula with $k$ free variables
and $\ell$ quantifiers.
Also, let $\mathcal{G}$ be any graph class of bounded expansion.
Then there exists a computable function $f:\NN\to\NN$, 
depending on~$\mathcal{G}$, such that
$$\xc\!\big(\FORM\nolimits_\phi(G)\big)\leqslant f(k+\ell)\cdot n $$ 
holds for every integer $n$ and every $n$-vertex graph~$G\in\mathcal{G}$.
Furthermore, an explicit extension of $\FORM\nolimits_\phi(G)$ of size at most
$f(k+\ell)\cdot n$ can be found in polynomial time for fixed $\mathcal{G}$ and $k,\ell$.
\end{theorem}

\begin{proof}
We set $q:=k+\ell$ and $i:=k$, and first apply Lemma~\ref{lem:fullFOtoex}
to obtain the labeling $G(c)$ of~$G$ (note that $G(c)$ has the same
underlying graph as~$G$ and so having the same bounded expansion) 
and the existential FO formulas $\psi_t$.
Let $\mathcal{T}'\subseteq\mathcal{T}_k^q$ be the subset of those $q$-types
which include our~$\phi$.
Then, by the definition of type and Lemma~\ref{lem:fullFOtoex}, we have 
$G\models\phi(\bar{v})$ for $\bar{v} \in V(G)^i$,
if and only if $G(c)\models \psi_t(\bar{v})$ for some $t\in\mathcal{T}'$.

Hence the polytope $\FORM\nolimits_\phi(G)$ is the convex hull of the union
of the polytopes $\FORM_{\psi_t}\!\big(G(c)\big)$ for $t\in\mathcal{T}'$.
Since the cardinality of $\mathcal{T}'\subseteq\mathcal{T}_k^q$ is finite and
bounded in terms of~$q=k+\ell$,
and the formulas $\psi_t$ depend only on $t\in\mathcal{T}_k^q$ 
and $q$ for a fixed class $\mathcal{G}$,
our result now directly follows from Lemma~\ref{lem:exFOupper},
applied to each $t\in\mathcal{T}'$, via Theorem~\ref{thm:union_balas}.
\end{proof}

\section{Nowhere Dense Classes}
\label{sec:nowheredense}

In this section we present yet another extension of
Theorem~\ref{thm:bdexptoxc}, studying the $k$-independent set polytope
(and more generally existential FO polytopes) on
graph classes larger than those with bounded expansion.

A graph class $\mathcal{G}$ is {\em nowhere dense}~\cite{NOdM12} 
if there is no integer $d$ such that $\mathcal{G}\gradd d$ contains all graphs.
Every graph class of bounded expansion is nowhere
dense, but the converse is not true.
The $k$-independent-set problem, and existential FO problems in greater
generality, are also known to be in FPT on every
nowhere dense class~\cite{NOdM12}, see also a more general result of~\cite{gks14}.
It is natural to ask whether the same can hold for the fixed-parameter
extension complexity of their polytopes.
Indeed, the following similarly holds true.

\begin{theorem}\label{thm:nodenseexFO}
Let $\phi(x_1,x_2,\dots,x_k)$ be an existential FO formula with $k$ free variables
and $\ell$ quantifiers.
Also, let $\mathcal{G}$ be any nowhere dense graph class. 
Then, for every $\varepsilon>0$,
there exists a computable function $f:\NN\to\NN$ depending on $\varepsilon$
and $\mathcal{G}$, such that
$$\xc\!\big(\FORM\nolimits_\phi(G)\big)\leqslant f(k+\ell)\cdot	n^{1+\varepsilon} $$ 
holds for every integer $n$ and every $n$-vertex graph~$G\in\mathcal{G}$.
Furthermore, an explicit extension of $\FORM\nolimits_\phi(G)$ of this size
can be found in polynomial time for fixed $\mathcal{G}$ and $k,\ell$.
\end{theorem}
\begin{proof}
  The approach for nowhere dense classes is nearly the same as in the proof
  of Lemma~\ref{lem:exFOupper}:  by~\cite{NOdM08,NOdM12},
  for a nowhere dense class~$\cal G$ and $\varepsilon' > 0$, $p \in \mathbb N$ there
  exists a threshold~$N_{\varepsilon',p}$ such that each~$G \in \cal G$ with
  $n=|G| \geq N_{\varepsilon',p}$ admits a low treedepth coloring
  of order $p$ with at most~$N=n^{\varepsilon'}$ colors.
  In such case we take $p:=k+\ell$ and $\varepsilon':=\varepsilon/(k+\ell)$.

  Now, setting 
  $\mathcal{J}_{k+\ell}:={[N] \choose k+\ell}$, we conclude as in
  the proof of Lemma~\ref{lem:exFOupper}:
\begin{align*}
\xc\!\big(\FORM\nolimits_\phi(G)\big) &\leqslant
	|\mathcal{J}_{k+\ell}|\cdot\big(1+g(k+\ell,k+\ell-1)\cdot n\big)
\\ &\leqslant N^{k+\ell}\cdot f(k+\ell)\cdot n
	= f(k+\ell)\cdot (n^{\varepsilon/(k+\ell)})^{k+\ell}\cdot n
\\	&= f(k+\ell)\cdot n^{1+\varepsilon}
\,.\end{align*}
On the other hand, for  $n=|G| < N_{\varepsilon',p}$ we have got a finite
problem which is solved by brute force.
\end{proof}

However, we now cannot directly proceed towards full FO logic in the direction
of Theorem~\ref{lem:fullFOupper} since we do not have a tool alike
Lemma~\ref{lem:fullFOtoex} available for the case of nowhere dense classes.

\section{Conclusions}
\label{sec:conclusions}

We have begun to study the question: to which extent FP tractability
of the $k$-independent set problem on graph classes is related to the FPT
extension complexity of the (corresponding) $k$-independent-set polytope?
Not surprisingly, we confirm that there cannot be FPT extensions of this
polytope in the class of all graphs (note, though, that our proof is
absolute and does not rely on the assumption $FPT\not=W[1]$).
On the other hand, the $k$-independent-set problem is linear-time FPT
on graph classes of bounded expansion~\cite{NOdM06}, and we construct a
linear FPT extension for its polytope on such classes.
This positive result then routinely carries over to all
FO problems on graph classes of bounded expansion.

We now outline possible natural directions of future research in this regard.
\begin{enumerate}
\item The deep tractability result of \cite{gks14}
addresses problems in full FO logic of graphs from nowhere dense classes.
This suggests that perhaps, for every FO formula~$\phi$
(not only the existential ones as in Theorem~\ref{thm:nodenseexFO}),
the related $\phi$-polytope can also have
FPT extension complexity on nowhere dense graph classes.
Though, the involved proof techniques of~\cite{gks14} do not seem to
easily translate to the extension complexity setting.
\end{enumerate}

In a broader view, one may regard the property of
a problem polytope having an~FPT extension complexity as a finer
(case-by-case) resolution of the class FPT.
For an explanation; the well-established assumption $FPT\not=W[1]$
implies that problems not in FPT do not have FPT extensions,
while on the other hand the example of the matching polytope \cite{Roth:14}
suggests that there may also be FPT problems whose polytopes do not have
FPT extensions.
(In other words, a situation could be analogous to that of polynomial
kernelization; while every problem with a polynomial kernel is FPT, many FPT
problems do not admit a polynomial kernel.)
We believe that this task is worth further detailed investigation.

One may try to proceed even further and ask a general question:
\begin{enumerate} 	
\setcounter{enumi}{1} \item \label{it:broad}
Is it true that all $W[t]$-hard problems for some $t\geq1$ do not admit
	FPT extensions?
\end{enumerate}
However,
this question is not even easy to formulate since the polytope we associate
with a problem remains a specific choice which by no means is the only choice.

Moreover, it can be argued that either possible answer to the very broad
Question \eqref{it:broad} would be a significant breakthrough in the
complexity world.
Say, since some FPT problems such as $k$-vertex cover do admit
FPT extensions \cite{B:15}, an affirmative answer to \eqref{it:broad} would imply 
that this problem is not $W[t]$-complete and so $FPT\neq W[t]$.
On the other hand, if the answer to \eqref{it:broad} was no, 
then this would imply the existence of non-uniform FPT circuits
for $W[t]$-complete problems which is considered unlikely.

%

\section*{References}
\bibliography{ght}

\end{document}